\documentclass{article}

\usepackage{amsthm,amsmath}
\usepackage{amssymb}

\usepackage[dvipdfm,bookmarks=true,bookmarksnumbered=true,%
bookmarkstype=toc]{hyperref}

\newtheorem{Def}{Definition}
\newtheorem{Thm}{Theorem}

\usepackage[longnamesfirst]{natbib}
\def\BibTeX{{\rm B\kern-.05em{\sc i\kern-.025em b}\kern-.08em
    T\kern-.1667em\lower.7ex\hbox{E}\kern-.125emX}}

\usepackage{indentfirst}

\begin{document}
\bibliographystyle{aer} 

\title{Third person enforcement in
a prisoner's dilemma game  }
\author{Tatsuhiro Shichijo}
\date{June 14, 2012}
\maketitle

\begin{abstract}
	We theoretically study the effect of a third person enforcement on a one-shot prisoner's dilemma game played by two persons, with whom the third person plays repeated prisoner's dilemma games. We find that the possibility of the third person's future punishment causes them to cooperate in the one-shot game.
\end{abstract}
\section{Introduction}
Three players play a specific repeated game, in which the stage game  is the prisoner's dilemma game illustrated in Table \ref{tbl:PD}, where $P=45, S=10, T=100, R=75$. 
 In the first stage, players $X_1$ and $X_2$ play the prisoner's dilemma game as the stage game. 
From the second stage, players $M$ and $X_1$ play the stage game with probability 1/2. Players $M$ and $X_2$ also play the stage game with probability 1/2. That is, player $M$ plays the stage game with certainty. Players $X_1$ and $X_2$ play the  stage game  with probability 1/2. They play the game an infinite number of times, with a discount factor of $\delta=3/4$. We assume that each player observes only the outcome of the stage game that s/he plays.
 For example, player $M$ cannot see the action profile of the first stage. 

Players $X_1$ and $X_2$ play the stage game against each other  only once. However, it is  possible that they both play $C$ in the first stage because a third person, player $M$, may enforce cooperation. 
 We analyze whether both players $X_1$ and $X_2$ play $C$ in a sequential equilibrium.

\cite{Kandori1992} showed that a contagious strategy constitutes a cooperative equilibrium in a private monitoring setting if the discount factor is sufficiently large. In section \ref{sec:kandori} of this paper, we show that the \cite{Kandori1992}-type  contagious strategy  cannot constitute a cooperative equilibrium under the parameter settings of this paper. In section \ref{sec:cooperative equilibrium}, we show that another type of strategy profile constitutes a sequential equilibrium. 

\begin{table}
\begin{center}
\begin{tabular}{|c|c|c|}
\hline
 & C & D \\
\hline
C & R & S \\
\hline
D & T & P \\
\hline
\end{tabular}
\caption{Prisoner's Dilemma Game}\label{tbl:PD}
\end{center}
\end{table}

\section{Notation}
 We denote by $(a_1 a_2 )$ the outcome of the first stage in which player $X_1$ plays $a_1$ and player $X_2$ plays $a_2$, where $a_1,a_2 \in \{C,D \}$.
 From the second stage, either player $X_1$ or player $X_2$ is chosen to play the stage game. To identify the selected player, we denote by $(X_i a_ia_M )$ the stage $t$ outcome in which player $X_i$ is selected and plays $a_i$ and player $M$ plays $a_M$. For example, $(X_1CD)$ is the stage outcome in which player $X_1$ plays $C$ and player $M$ plays $D$.
 We denote by $(a^1_1a^1_2;  X_ia^2_ia^2_M ; \dots; X_ja^t_ja^t_M)$ the history of the outcome up to stage $t$.
 Let $H^t$ be the set of histories up to stage $t$.
 The behavioral strategy of player $X_i$ at stage $t$  depends on the history up to stage $t-1$. The behavioral strategy at stage $t$ of player $X_i$ is described by the function $\sigma^t_i:H^{t-1}  \rightarrow \{C,D \}$. 
 By contrast, the behavioral strategy of player $M$ depends on who the opponent is. The behavioral strategy of player $M$ is described by the function $\sigma^t_M:H^{t-1} \times \{X_1,X_2 \} \rightarrow \{C,D \}$.
 When we do not specify player, we use $z$. For example, $(X_zDD;X_zDD)$ means that one of the following outcomes occurs: $(X_1DD;X_1DD)$,
 $(X_1DD;X_2DD)$, $(X_2DD;X_1DD)$, $(X_2DD;X_2DD)$. 
 When we do not specify action, we use $Z$. For example, $(X_1ZZ)$ means that one of the following outcomes occurs: $(X_1CC), (X_1CD), (X_1DC), (X_1DD)$.
 We denote the sequence of $\sigma^t_i$ by $\sigma_i$; i.e., $\sigma_i=(\sigma^1_i, \sigma^2_i,\dots)$. We denote the sequence of $\sigma^t_M$ by $\sigma_M$. Let $\sigma=(\sigma_1,\sigma_2,\sigma_M)$.

\section{The contagious strategy}\label{sec:kandori}
In this section,  we show that a \cite{Kandori1992}-type contagious strategy cannot constitute a sequential equilibrium under the parameter settings of this paper.
A player who plays a contagious strategy plays $D$ if her/his opponent has previously played  $D$ against her/him.
For example, if player $X_1$ plays $D$ against player $M$, then player $M$ plays $D$ not only against player $X_1$, but also against player $X_2$.
If s/he has previously played $D$ against a player, then s/he again plays $D$ against that player. \footnote{Because our game setting is different from that of \cite{Kandori1992}, the behavioral strategy is slightly different from \cite{Kandori1992}'s contagious strategy. That is, even if s/he has played $D$ against player $X_1$, s/he plays $C$ against player $X_2$.} For example, if player $M$ plays $D$ against player $X_1$, then player $M$ uses strategy $D$ against player $X_1$. Otherwise, s/he plays $C$.

When player $M$ observes a deviation by her/his opponent, s/he assumes that the deviation occurred in the first stage if it is a reasonable deviation. Suppose player $M$ observed $(X_1CC;X_2DC)$ in the second and third stages. There are two explanations: (i) player $X_1$ or player $X_2$ played $D$ in the first stage, but player $X_1$ played $C$ even though $X_1$ was supposed to play $D$; (ii) there was no deviation in the first or second stage, but player $X_2$ deviates in the third stage for the first time.
We assume that player $M$ follows (i) and that the player $M$ uses $D$ against both other players.

We can constitute a sequence of assessments with completely mixed strategies that is consistent with the contagious strategy profile and with beliefs that satisfy the above principle. 
Let $\gamma$ be the mixed strategy such that strategy $C$ is played with probability 1/2 and strategy $D$ is played with probability 1/2. 
Let $\hat{\sigma}$ be the contagious strategy. We can base a complete mixed strategy $\tilde{\sigma}$ on $\hat{\sigma}$ with $\epsilon>0$ as follows:
\begin{gather*}
\tilde{\sigma}^1_i(\emptyset)=(1-\epsilon)\hat{\sigma}^1_i(\emptyset) + \epsilon \gamma \text{ for } i \in \{1,2 \}\\
\tilde{\sigma}^t_i(CD;\dots) =(1-\epsilon)\hat{\sigma}^t_i(CD;\dots) + \epsilon \gamma \text{ for } i \in \{1,2 \} , t>1\\
\tilde{\sigma}^t_i(DC;\dots) =(1-\epsilon)\hat{\sigma}^t_i(DC;\dots) + \epsilon \gamma \text{ for } i \in \{1,2 \} , t>1\\
\tilde{\sigma}^t_M(ZZ;\dots \mid X_i) =(1-\epsilon)\hat{\sigma}^t_M(ZZ;\dots \mid X_i) + \epsilon \gamma \text{ for } i \in \{1,2 \} , t>1\\
\tilde{\sigma}^t_i(CC;\dots) = (1-\epsilon^{1/\epsilon})\hat{\sigma}_i(CC;\dots) + \epsilon^{1/\epsilon} \gamma \text{ for } i \in \{1,2 \} , t>1.
\end{gather*}
 We can base the belief $\mu_\epsilon$ on $\tilde{\sigma}$ by Bayes' rule. 
Taking the limit as $\epsilon \rightarrow 0$, $\tilde{\sigma}$ converges to $\hat{\sigma}$ and $\mu_\epsilon$ converges to a belief that satisfies the above principle. This is because $\lim_{\epsilon \rightarrow 0} \epsilon^{1/\epsilon} / \epsilon^k=0$ for all $k \in \mathbb{N}$.

The payoff of player $X_i$ from the contagious strategy profile is 
$R + \delta R/(2(1-\delta))$.
If player $X_i$ plays $D$ in every stage, her/his payoff is $T+\delta T/2 + \delta^2 P/(2(1-\delta))$. If $\delta \geq 0.752903$,
$R + \delta R/(2(1-\delta)) \geq T+\delta T/2 + \delta^2 P/(2(1-\delta))$. If $\delta = 3/4 = 0.75$, which is the parameter setting in this paper, the contagious strategy profile cannot be a sequential equilibrium.

\section{A cooperative equilibrium}\label{sec:cooperative equilibrium}

In this section, we consider a variation of the contagious strategy and show that the  new strategy profile $\sigma$ constitutes a sequential equilibrium. As with the contagious strategy, this strategy is to play $D$ forever if s/he observed that one of his/her opponents  deviated from the strategy. For example, player $M$ plays $D$ against $X_1$ if player $M$ observed that player $X_2$ deviated from the strategy profile.  
The difference between our strategy and the contagious strategy relates to the behavioral strategy in the third stage. If player $X_i$ is selected in the second stage and player $X_j(j \neq i)$ is selected in the third stage, then player $X_j$ is allowed to play strategy $D$ in the third stage . In this case, the outcome in the third stage is $(X_jDC)$. Thereafter, player $M$ and player $X_j$ continue to choose $C$. On the other hand, if player $X_i$ is selected in the second and third stages, then player $X_i$ must play $C$ in the third stage. In this case, if player $X_i$ plays $D$ in the third stage, then player $M$ plays $D$ thereafter.

For example, $(CC;X_1CC;X_1CC;X_zCC;X_zCC;X_zCC;\dots)$ or \\
 $(CC;X_1CC;X_2DC;X_zCC;X_zCC;\dots)$ are outcomes on the path of the strategy.

Formal definitions of $\sigma$ are as follows:
\begin{Def}
\begin{flalign*}
&\sigma^1_1(\emptyset)=\sigma^1_2(\emptyset)=C&\\
&\sigma^2_1(CC)=\sigma^2_2(CC)=C&\\
&\sigma^2_M(ZZ \mid X_i)=C \text{ for $i \in \{1,2 \}$}&\\
&\sigma^3_i(CC;X_iCC)=C \text{ for } i=\{1,2 \}&\\
&\sigma^3_i(CC;X_jCC)=D \text{ for $i,j=\{1,2 \}$, where } j \neq i&\\
&\sigma^3_M(ZZ;X_zCC \mid X_z)=C&\\ 
&\sigma^3_M(ZZ;X_iCD \mid X_j)=C \text{ for $i,j=\{1,2 \}$, where } j \neq i&\\ 
&\sigma^4_i(CC;X_iCC;X_iCC ) =C \text{ for $i,j=\{1,2 \}$, where }i \neq j&\\
&\sigma^4_i(CC;X_iCC;X_jZZ ) =C \text{ for $i,j=\{1,2 \}$, where }i \neq j&\\
&\sigma^4_i(CC;X_jZZ;X_iDC ) =C \text{ for $i,j=\{1,2 \}$, where }i \neq j&\\
&\sigma^4_i(CC;X_jZZ;X_jZZ ) =C \text{ for $i,j=\{1,2 \}$, where }i \neq j&\\
&\sigma^4_M(ZZ;X_iCC;X_iCC \mid X_i) =C \text{ for $i=\{1,2 \}$} &\\
&\sigma^4_M(ZZ;X_iCZ;X_iCZ \mid X_j) =C \text{ for $i,j=\{1,2 \}$, where }i \neq j&\\
&\sigma^4_M(ZZ;X_iCC;X_jZZ \mid X_i) =C \text{ for $i,j=\{1,2 \}$, where }i \neq j&\\
&\sigma^4_M(ZZ;X_iCZ;X_jDC \mid X_j) =C \text{ for $i,j=\{1,2 \}$, where }i \neq j&.
\end{flalign*}
The behavioral strategy played up to stage 4, which is not listed above, is $D$.
From the fifth stage, the strategy is the same as the contagious strategy. That is, if a player plays $D$ after the fifth stage, her/his opponent subsequently plays  $D$ and s/he subsequently plays $D$ against the opponent.
\end{Def}

When player $M$ observes a deviation by her/his opponent, as in section \ref{sec:kandori}, player $M$ presumes that this deviation occurred in the first stage if it is reasonable.  On the other hand, if player $M$ observes $(X_1CC;X_2CC)$, s/he does not assume that the deviation occurred in the first stage because it is unreasonable. As in section \ref{sec:kandori}, a belief that satisfies the above principle is the limit of the beliefs based on the complete mixed strategy.

We show that the above strategy profile constitutes  a sequential equilibrium for $\delta = 0.75$.
\begin{Thm}
$\sigma$ constitutes a sequential equilibrium if $\delta=0.75$, $P=45, S=10, T=100, R=75$.
\end{Thm}
\begin{proof}
We investigate the following cases.
\begin{description}
\item[Case 1]  (in which  $(X_iDD)$ is assumed to be played in the strategy profile $\sigma$): 
The stage payoff obtained from playing $C$ is lower than that obtained from playing $D$. Regardless of the action s/he takes, the opponent continues to play $D$ in subsequent stages. Playing $C$ never improves the payoff  obtained from the next stage. Thus, there is no incentive to deviate.
\item[Case 2] (in which there is no deviation and in which $(X_iCC)$ is assumed to be played in the strategy profile in the second stage or later):
 The expected continuation payoff obtained by player $X_i$ from playing $C$ is $R + \delta R /(2(1-\delta))=187.5$. The expected continuation payoff obtained by player $X_i$ from playing $D$ is, at most, $T + \delta P / (2(1-\delta))=167.5$. Thus, player $X_i$ has no incentive to deviate. It is easily checked that the same applies for player $M$.
\item[Case 3] (in which, after playing $(CC;X_iCC)$, player $X_j(j\neq i)$ is selected to play in the third stage): 
Clearly, player $X_j$ has no incentive to deviate because player $X_j$ is expected to play $D$. Player $M$'s expected continuation payoff from playing $C$ is $S + \delta R /(1-\delta)=  235$. Player $M$'s expected continuation payoff from playing $D$ is, at most, $P + \delta/2 \times (P+ R) /(1-\delta)=225$. Thus, there is no incentive to deviate.
\item[Case 4] (in which player $M$ gets to choose an alternative in the second stage):  
The expected continuation payoff for player $M$ from playing $C$ is $R + \delta R /2 + \delta S / 2 + \delta^2 R/(1-\delta)=275.625$. The expected continuation payoff for player $M$ from playing $D$ is, at most, $T + \delta (S/2 + P/2) + \delta P /(2(1-\delta))+  \delta R /(2(1-\delta))=255.625$. Thus, there is no incentive to deviate.
\item[Case 5] (in which the first stage outcome  is $(CD)$ and the current stage outcome is assumed to be $(X_iDC)$):
 The expected continuation payoff for player $X_i$ from playing $D$ is $T+\delta P/(2(1-\delta))=167.5$. The payoff obtained from playing $C$ is, at most, $R+\delta T/2 + \delta^2 P/(2(1-\delta))=163.125$. Note that player $M$ adopts a type of contagious strategy. Player $M$ plays $D$ against $X_i$ after player $M$ plays against $X_j(j\neq i)$.
If player $X_i$ continues to play $C$ whenever player $M$ does not play against $X_j(\neq  X_i)$, the expected continuation payoff is: 
$$R+ \delta \frac{ R}{2} + \frac{1}{2} \sum_{s=2} \delta^s \left\{ \left(\frac{1}{2}\right)^{s-1} R + (1-\left(\frac{1}{2}\right)^{s-1} ) P \right\}$$ 
$$=R+ \delta \frac{R}{2}+\frac{1}{2} \left( \frac{\delta^2 (R - P)}{2-\delta} + \frac{\delta^2 P}{1-\delta}  \right) =160.5.$$ Thus, there is no incentive to deviate in this case.
\item[Case 6] (in which player $X_i$ gets to choose an alternative in the first stage): 
If player $X_i$ plays $C$, the expected continuation payoff is $R + \delta R/2 + \delta^2 ( R + T)/4 + \delta^3 R/(2(1-\delta))=191.016 $. If player $X_i$ plays $D$ in the first stage, the expected continuation payoff is, at most, $T+ \delta T/2 + \delta^2 P/(2(1-\delta))=188.125 $.  Thus, there is no incentive to deviate in this case.
\end{description}

The above results show that our strategy profile constitutes a sequential equilibrium strategy.

\end{proof}

\begin{verbatim}
sha256 9153673a330a7c8d01edb860ca8ffdd04be176cf0cc0245be495d27cfae488d0  
19001131-001-40305660-07-0001.zip
0f6a0c1477c2a8c70cf1d83edaff0512dc85e699e2afd8794d2328130bf16c66  
GAIYO_19001131-001-40305660-07-0001.pdf
ce92efa636ce53fa0fcaada157a5334e3f79c8f0865f2c0d704aff52d09da618  
MoneyDraft.pdf
\end{verbatim}


\end{document}